\documentclass[a4paper,onecolumn,10pt]{quantumarticle}
\pdfoutput=1
\usepackage[utf8]{inputenc}
\usepackage[english]{babel}
\usepackage[T1]{fontenc}
\usepackage{amsmath}
\usepackage[shortlabels]{enumitem}
\usepackage{hyperref}
\usepackage{tikz}
\usepackage{float}
\usepackage{lipsum}
\usepackage{times,amsmath,amsfonts,amssymb,latexsym,amsthm}
\usepackage{graphicx,epsf}
\usepackage{subfigure}
\usepackage[toc,page,header]{appendix}
\usepackage{minitoc}
\usepackage[numbers,sort&compress]{natbib}
\usepackage{color}
\usepackage{setspace}
\newcommand{\be}{\begin{equation}}
\newcommand{\ee}{\end{equation}}
\newcommand{\ba}{\begin{eqnarray}}
\newcommand{\ea}{\end{eqnarray}}
\newcommand\tr{{\operatorname{tr}}}
\newcommand{\ignore}[1]{}

\newcommand{\ket}[1]{\left | {#1} \right \rangle }

\newcommand{\bra}[1]{\left \langle {#1} \right | }

\newcommand{\de}[0]{{\operatorname{d}}}

\newcommand{\aver}[1]{ \langle  {#1}  \rangle }

\def\CC{{\rm\kern.24em \vrule width.04em height1.46ex depth-.07ex
    \kern-.29em C}}
\def\P{{\rm I\kern-.25em P}}
\def\RR{{\rm
         \vrule width.04em height1.58ex depth-.0ex
         \kern-.04em R}}
\def\bbbone{{\mathchoice {\rm 1\mskip-4mu l} {\rm 1\mskip-4mu l}
{\rm 1\mskip-4.5mu l} {\rm 1\mskip-5mu l}}}
\def\bbbc{{\mathchoice {\setbox0=\hbox{$\displaystyle\rm C$}\hbox{\hbox
to0pt{\kern0.4\wd0\vrule height0.9\ht0\hss}\box0}}
{\setbox0=\hbox{$\textstyle\rm C$}\hbox{\hbox
to0pt{\kern0.4\wd0\vrule height0.9\ht0\hss}\box0}}
{\setbox0=\hbox{$\scriptstyle\rm C$}\hbox{\hbox
to0pt{\kern0.4\wd0\vrule height0.9\ht0\hss}\box0}}
{\setbox0=\hbox{$\scriptscriptstyle\rm C$}\hbox{\hbox
to0pt{\kern0.4\wd0\vrule height0.9\ht0\hss}\box0}}}}
\def\bbbz{{\mathchoice {\hbox{$\sf\textstyle Z\kern-0.4em Z$}}
{\hbox{$\sf\textstyle Z\kern-0.4em Z$}}
{\hbox{$\sf\scriptstyle Z\kern-0.3em Z$}}
{\hbox{$\sf\scriptscriptstyle Z\kern-0.2em Z$}}}}

\renewenvironment{proof}{{\bf \emph{Proof.} }}{\hfill $\Box$ \\} 

\usepackage{hyperref}

\newtheorem*{theorem*}{Theorem}
\newtheorem{defin}{Definition}
\newtheorem*{defin*}{Definition}

\newtheorem*{corollary*}{Corollary}
\newtheorem{lemma}{Lemma}
\newtheorem*{lemma*}{Lemma}

\newtheorem*{remark*}{Remark}

\begin{document}
\setcounter{secnumdepth}{3}
%\onecolumngrid

 %%%%%%%%%%%%%%%%%%%%%%%%%%%%%%%%%%%%%%%%%%%%%%%%%%%%%%%%%%%%%%%

\title{Stability of topological purity under random local unitaries}

\author{Salvatore F.E. Oliviero}\email{s.oliviero001@umb.edu}
\affiliation{Physics Department,  University of Massachusetts Boston,  02125, USA}
\author{Lorenzo Leone}
\affiliation{Physics Department,  University of Massachusetts Boston,  02125, USA}
\author{You Zhou}
\affiliation{Centre for Quantum Technologies. National University of Singapore, 3 Science Drive 2, Singapore 117543, Singapore}
\affiliation{Nanyang Quantum Hub, School of Physical and Mathematical Sciences, Nanyang Technological University, Singapore 637371, Singapore}
\author{Alioscia Hamma}
\affiliation{Physics Department,  University of Massachusetts Boston,  02125, USA}
\affiliation{Université Grenoble Alpes, CNRS, LPMMC, 38000 Grenoble, France}

\begin{abstract}

In this work, we provide an analytical proof of the robustness of topological entanglement under a model of random local perturbations. We define a notion of average topological subsystem purity and show that, in the context of quantum double models, this quantity does detect topological order and is robust under the action of a random quantum circuit of shallow depth. 
\end{abstract}

\maketitle
\section*{Introduction}
Topological order\cite{Wen2007Quantum} is a novel kind of quantum order that goes beyond the paradigm of symmetry breaking. Its role is prominent in condensed matter theory as well as in quantum computation. In particular, topological order can be employed to construct various models for robust quantum memory and logic gates\cite{Kitaev2003fault,key1943131m}. Topologically ordered states show patterns of non local quantum entanglement that cannot be detected by a local order parameter. However, the long-range quantum entanglement leaves its mark in the reduced density matrix, and a series of papers have shown that topological order can be detected by the topological entropy\cite{Hamma2005ground,Kitaev2006topological, Levin2006topological}: a topological correction to the area law for the entanglement entropy. In particular, topological entanglement entropy has been employed to characterize the ground state of different models\cite{Chung2010topo,Mezzacapo12ground,Fradkin2007topo,Zhang2011Topo,Jiang2012heisenberg,Jiang2012topological,Furukawa2007topo,Isakov2011topo,Micallo2020topo,Gong2013honey,Dong_2008topo}   Recent works have shown that this type of long-range quantum entanglement together with the topological entropy is robust against local perturbation of the Hamiltonian\cite{Trebst2007topo, Jadamagni2018robust, Dusuel2011robust} and small deformation of partition geometry. Attempts at showing the robustness of the topological order under local perturbations\cite{Bravyi2010topological,Bravyi2011topo} are either a proof of the stability of the phase or are often limited to specific examples and are mostly of numerical nature or resort to quantum field theory arguments\cite{Kitaev2006topological}, while an analytical proof for the robustness of topological entropy in lattice models is still lacking\cite{Jahromi2013robust, Hamma2013local,Hamma2014local}. In this work, we provide an analytic proof of the robustness of topological order under a noise model consisting of random local unitaries. To this end, we construct a notion of average topological subsystem purity that captures the same long-range pattern of entanglement of topological entanglement entropy, and we show that such topological purity is constant if the circuit is shallow compared to the relevant size of the subsystem (which, can be made scale with the size of the whole system $N$).

We work in the framework of quantum double models on the cyclic group $\mathbb{Z}_{d}$ introduced by Kitaev in \cite{Kitaev2003fault}, and define the topological purity (TP), which is related to the topological $2-$R\'enyi entropy defined in \cite{Flammia2009topological}. There are many reasons to use purity instead of entanglement entropy in order to argue about questions about quantum many-body systems. Unlike  the Von Neumann entanglement entropy (whose measurement requires a complete state tomography of the system\cite{Amico2008Ent}), the
$2-$R\'enyi entropy is directly related to the purity which is an observable and can be measured directly\cite{Horodecki2002method, Hastings2010Renyi, Enk2012measure, Abanin2012meas,Brydges2019randomized} as it is the expectation value of the swap operator over two copies of the system.  This quantity contains substantial information about the universal properties of quantum many-body systems\cite{Zanardi2013Ent} and it is able to reveal the topological pattern of entanglement\cite{Flammia2009topological,Hamma2012topo}. This property makes purity also amenable analytical treatment\cite{page1993average,Bose2001pur,Zanardi2004pur,LIDAR200682,DePasquale2011stat,streltsov2018coherence}.

To prove the robustness of topological order by the topological purity we introduce, as a noise model, a set of quantum maps whose action on a state is based on local random (shallow) quantum circuits.  We find that the topological purity distinguishes two phases of states, attaining two different constant values. When the circuit depth is comparable with the subsystem size, the long-range pattern of entanglement that is responsible for topological order can be changed and the topological purity can change value. The phase is then indeed the orbit of the so defined set of quantum maps through a reference state. The proof is obtained thanks to two key non trivial facts: (i) the subsystem purity of the ground state of $\mathbb{Z}_{d}$ quantum double models only depends on the geometry of its boundary, while the topological purity only depends on its topology, and  (ii) the action of the specific noise model we work with can be regarded as the evolution of that boundary. Since the maps are shallow, their action will result in a local deformation of the subsystem boundary that does not alter their topology, and, by (i), this will result in an exactly constant topological purity. Similarly, we show that the topological purity of a topologically trivial state is zero and that it cannot be changed by our noise model.

The paper is organized as follows: in Sec.\ref{sec:QuantumDouble} we review $\mathbb{Z}_d$ quantum double models; in Sec.\ref{sec:PT} we introduce the topological purity and discuss how it is connected with other measures of topological entropy; in Sec.\ref{phases} we introduce the noise model and finally in Secs.\ref{model} and \ref{mainresult}, will be devoted to the rigorous proof of our result and will be rather technical. 

\section{Quantum Double models on $\mathbb{Z}_{d}$}\label{sec:QuantumDouble}
Quantum double models are exactly solvable models defined on a lattice\cite{Kitaev2003fault}. Consider the cyclic finite group $\mathbb{Z}_{d}$ with $|\mathbb{Z}_{d}|=d$ and local Hilbert spaces $\mathcal{H}_{i}\simeq \mathbb{C}^{d}$ and the total Hilbert space given by the tensor product of $N$ local Hilbert spaces, namely $\mathcal{H}=\bigotimes_{i=1}^{N}\mathcal{H}_{i}$ placed at  the bonds of a square lattice $(V,E)$, see  Fig.\ref{lattice}. The dimension of the total Hilbert space is thus $D=d^N$. Let  $B\equiv\{\ket{n}|\, n=0,\dots,  d-1\}$ be an orthonormal basis in $\mathcal{H}_{i}\simeq \mathbb{C}^{d}$.
\begin{figure}[h]
    \centering
    \resizebox{6cm}{!}{\includegraphics[scale=0.5]{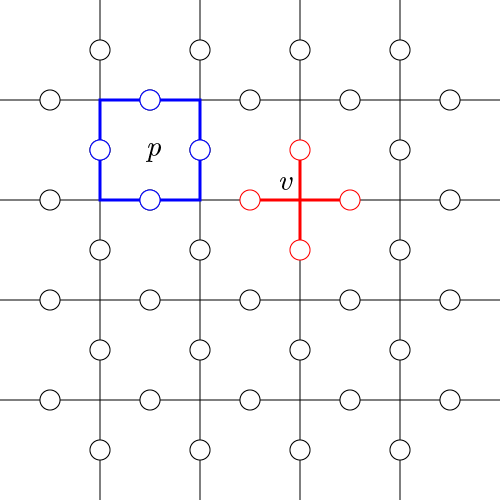}}
    \caption{A system of spins on square lattice, plaquette and star are denoted respectively by $p$ and $v$.}
    \label{lattice}
\end{figure}
For each local Hilbert space $\mathcal{H}_{i}$ we introduce the operators $\tilde{X}, \tilde{T}^{(j)}$ defined through their action on the ket $\ket{n}$:
\be
\tilde{X}^{m}\ket{n}=\ket{n+m}, \quad \tilde{T}^{(m)}\ket{n}=\delta_{mn}\ket{n}
\ee
where $\tilde{X}^{m}:=XX\cdots X$ $m$ times and the addition is modulo $d$. Consider the enlarged operators $X_{i}:=\tilde{X}_{i}\otimes \bbbone_{\mathcal{H}\setminus i}$ and $T^{(m)}_{i}:=\tilde{T}^{(m)}_{i}\otimes \bbbone_{\mathcal{H}\setminus i}$ acting non trivially only on the site $i\in V$. Define the following operators acting non trivially on the subset $v\subset V$, sketched in Fig.\ref{lattice}:
\be
A_{m}(v)=\prod_{i\in v} X_{i}^{m}, \quad B(p)=\sum_{\substack{m_{1},m_{2},m_{3},m_{4}\\ m_{1}+m_{2}+m_{3}+m_{4}=0, \mod{d}}}T^{(m_1)}_{i_1}T^{(m_2)}_{i_2}T^{(m_3)}_{i_3}T^{(m_4)}_{i_4}
\ee
note that $B(p)$(plaquette operator) and $A(v)=d^{-1}\sum_{m=0}^{d-1}A_{m}(v)$(star operator) are projectors. At this point, the Hamiltonian of the quantum double model reads:
\be
H_{QD}=\sum_{v}(\bbbone-A(v))+\sum_{p}(\bbbone-B(p))
\ee
and the ground state manifold $\mathcal{L}$ is given by:
\be
\mathcal{L}=\{\ket{\psi}\in\mathcal{H}|\, A(v)\ket{\psi}=\ket{\psi}, \, B(p)\ket{\psi}=\ket{\psi}\}.
\label{manifold}
\ee
 To represent the ground state in terms of the spin degrees of freedom, let us introduce $G$  the group generated by all the $A_{m}(v)$ operators, defined as $G=\langle\{A_{m}(v)\,|m=0,\dots, d-1,\,v=1,\dots, N/2\}\rangle$. The  state $\ket{\psi_{GS}}$  defined as 
\be
\ket{\psi_{GS}}=\prod_{s}A(s)\ket{0}^{\otimes N}=d^{-N/2}\prod_{s}\sum_{m=0}^{d-1}A_{m}(s)\ket{0}^{\otimes N}=d^{-N/2}\sum_{h\in G}h\ket{0}^{\otimes N}
\ee
is a state in $\mathcal{L}$, as it can be readily checked. Other basis states in $\mathcal{L}$ can be constructed by the use of non contractible loop operators\cite{Kitaev2003fault}. 
The topological order in this model can be detected by  the entanglement entropy in the ground state manifold. Consider a bipartition in the Hilbert space, namely $\mathcal{H}=\mathcal{H}_{\Lambda}\otimes \mathcal{H}_{\bar{\Lambda}}$ and compute the reduced density matrix $\rho_{\Lambda}$\cite{Flammia2009topological}:
\be
\rho_{\Lambda}=\tr_{\bar{\Lambda}}\Psi_0=\frac{|G_{{\bar{\Lambda}}}|}{|G|}\sum_{h\in G/G_{{\bar{\Lambda}}},\tilde{h}\in G_{\Lambda}}h_{\Lambda}^{-1}\ket{0}\bra{0}^{\otimes N}h_{\Lambda}\tilde{h}_{\Lambda}
\ee
where $\Psi_0\equiv \ket{\psi_{GS}}\bra{\psi_{GS}}$ and we introduced $G_{\Lambda}:=\{g \in G|\, g=g_{\Lambda}\otimes \bbbone_{\bar{\Lambda}}\}$ and $G_{\bar{\Lambda}}:=\{g \in G|\, g=\bbbone_{\Lambda}\otimes g_{\bar{\Lambda}}\}$ that are normal groups in $G$, and the quotient groups $G/G_{\Lambda}$ and $G/G_{\bar{\Lambda}}$. Following \cite{Flammia2009topological} we can prove that $\rho_{\Lambda}^{2}=\frac{|G_{\Lambda}||G_{\bar{\Lambda}}|}{|G|}\rho_\Lambda$ and thus the purity is given by $P_{\Lambda}(\rho)=\frac{|G_{\Lambda}||G_{\bar{\Lambda}}|}{|G|}$, i.e one can argue that the purity is counting the number of \textit{independent} operators $A_{m}(v)$ acting non trivially on both regions $\Lambda$ and $\bar{\Lambda}$. Following \cite{Hamma2005ground}, given a region $\Lambda$, the number of $A_{m}(v)$ operators acting on both subsystems $\Lambda$ and $\bar{\Lambda}$ is $d^{|\partial \Lambda|- n_{2}-2n_{3}}$ where $|\partial\Lambda|$ is the cardinality of the boundary of $\Lambda$, i.e the number of sites in $\bar{\Lambda}$ having at least one nearest neighbor inside $\Lambda$, and $n_{i}$, for $i=2,3$, is the number of sites in $\bar{\Lambda}$ having $i$ nearest neighbors inside $\Lambda$. Thus $n_{2}+2n_{3}$ is a geometrical correction which depends on the shape of the region $\Lambda$. For example, if $\Lambda$ is a convex loop (a rectangle) $n_{2}=n_{3}=0$. So far we accounted for the number of star operators acting on both subsystems, but not all of them are independent from each other because of the constraints on the ground state manifold in \eqref{manifold}, in particular the condition $\ket{\psi_{GS}}\in\mathcal{L}\iff \prod_{p}B(p)\ket{\psi_{GS}}=\ket{\psi_{GS}}$. Following \cite{Hamma2005ground, Levin2006topological} and defining $n_{\partial}(\Lambda)$ as the number of boundaries of $\Lambda$, we have that the number of independent star operators is $d^{|\partial \Lambda|- n_{2}-2n_{3}-n_{\partial}(\Lambda)}$, i.e for each boundary of $\Lambda$ we have that the number of independent star operators acting on both subsystems decreases of a factor scaling as $d^{-1}$. We thus can finally write the following:
\be
P_{\Lambda}(\Psi_0)=2^{-\log_{2}d|\partial\Lambda|+\Gamma_{\Lambda}}
\label{purity}
\ee
where $\Gamma_{\Lambda}=\gamma_{\Lambda}+n_{\partial}(\Lambda)\gamma$ is the sum of a geometrical term $\gamma_{\Lambda}=\log_2 d(n_2+2n_3)$ which depends on the shape of the boundary $\partial\Lambda$ and a topological correction $n_{\partial}(\Lambda)\gamma$, due to the actual number of independent star operators, only related to the topology of $\Lambda$. This topological correction $\gamma\equiv \log_{2}d$ is called \textit{topological entropy}\cite{Hamma2005ground,Levin2006topological,Kitaev2006topological}. Eq.\eqref{purity} is of fundamental importance for the reminder of the paper: it is telling us that the purity of the reduced density matrix in the ground state manifold of the topologically ordered quantum double model depends on the boundary $\partial \Lambda$ only. 
\section{The topological purity}\label{sec:PT}
In this section, we show how the topological pattern of entanglement involved in topologically ordered states\cite{Wen2007Quantum} can be also found in a new quantity: the topological purity (TP). To understand heuristically how this quantity works, let us first introduce the topological entropy: consider the state $\sigma$ living in the Hilbert space $\mathcal{H}\equiv\mathcal{H}_A\otimes \mathcal{H}_B\otimes \mathcal{H}_C\otimes \mathcal{H}_{D}$ and the regions $AB,BC, B$ and $ABC$ drawn in Fig.\ref{fig1} $(a)$. The topological entropy is defined as 
\be 
S_{top}(\sigma)=S_{ABC}(\sigma)+S_{B}(\sigma)-S_{AB}(\sigma)-S_{BC}(\sigma)
\label{topentropy}
\ee 
where $S_{\Lambda}(\sigma)$ labels the Von Neumann entropy of $\tr_{\bar{\Lambda}}(\sigma)$ where $\bar{\Lambda}$ is the complement of $\Lambda$ with respect to $ABCD$. As it was shown in\cite{Flammia2009topological}, also all the Renyi topological entropies give exactly the same results for the quantum double models. The definition of the topological entropy is equal to minus the quantum conditional information $I(A;C|B)$\cite{Nielsen}, which is a entropic quantity describing tripartite correlations of quantum states. 
\begin{figure}[H]
\centering
\resizebox{9cm}{!}{\includegraphics[scale=0.8]{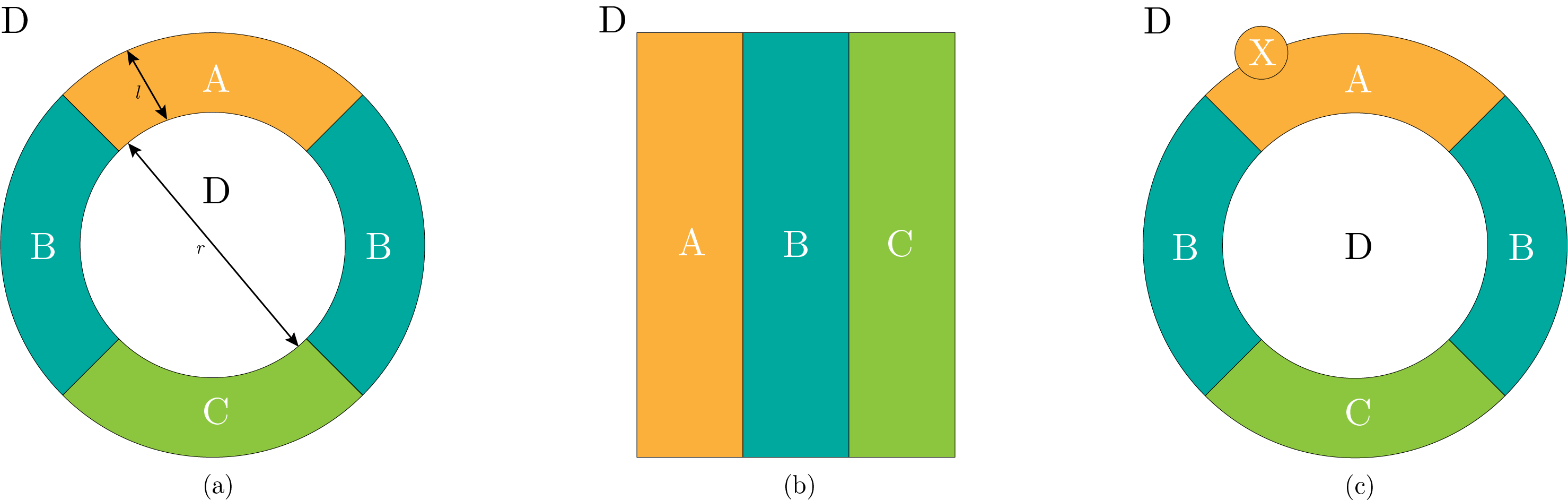}}
\caption{$(a)$ The graph configuration to define the topological entropy, i.e. $I(A;C|B)$. $l$ is the feature size of this graph configuration and $l\sim O(L)$, where L is the whole lattice size, while $r$ is radius of this topologically non trivial domain. $(b)$ $I(A;C|B)=0$ in this simple graph configuration, even though there is long-range entanglement of the ground state. $(c)$ Local boundary modification that does not change the topology of $ABC$; as proven in Lemma \ref{lemma}, the topological purity computed in $ABC$ equals the topological purity computed in $ABC\cup X$.}\label{fig1}
\end{figure}

In the same fashion of Eq.\eqref{topentropy}, the topological purity is defined as:
\begin{equation}\label{Eq:toppurity}
P_{top}(\sigma):=\frac{P_{AB}(\sigma)P_{BC}(\sigma)}{P_{B}(\sigma)P_{ABC}(\sigma)}
\end{equation}
i.e. the ratio of purities of the reduced density matrix of a quantum state $\sigma$ in the four subsystems $AB$, $BC$, $B$ and $ABC$. The purity of $\sigma$ in the subsystem $\Lambda$ is defined as $P_{\Lambda}(\sigma):=\tr [(\tr_{\bar{\Lambda}}\sigma)^{2}]$, where $\bar{\Lambda}$ is the complement of $\Lambda$. 
By definition,
\begin{equation}\label{Eq.logtoppurity}
-\log P_{top}(\sigma)=-\log P_{AB}(\sigma)- \log P_{BC}(\sigma)+\log P_{B}(\sigma)+\log P_{ABC}(\sigma)\\
\end{equation}
is just the topological 2-R\'enyi entropy.
Consider the domain $ABC$ in Fig.\ref{fig1} $(a)$, then $P_{\Lambda}=2^{-\log_2 d|\partial\Lambda|+\Gamma_{\Lambda}}$ for each $\Lambda\in\{AB,BC,B,ABC\}$ and since $|\partial AB|+|\partial BC|=|\partial B|+|\partial ABC|$, the topological purity is just given by the sum of the geometrical corrections 
\be
P_{top}(\Psi_{0})=2^{\Gamma_{AB}+\Gamma_{BC}-\Gamma_{B}-\Gamma_{ABC}}\equiv 2^{-2\gamma}
\label{toprule}
\ee
where $-2\gamma\equiv \Gamma_{AB}+\Gamma_{BC}-\Gamma_{B}-\Gamma_{ABC}$ is the topological entropy\cite{Levin2006topological}. Note that, according to the discussion in the previous section, all the geometrical corrections related to the shape of the boundary $\partial (ABC)$ are canceled by the choice of the partitions $AB, BC, B, ABC$, namely $\gamma_{AB}+\gamma_{BC}=\gamma_{B}+\gamma_{ABC}$, and the only surviving term is the topological correction that does not depend on the shape of the boundary: it is a purely topological correction $\propto \gamma$. This correction is the mark of the topological phase. It is worth noting that the topological correction would not be detected from the topological purity if $ABC$ was a simply connected region as the one sketched in Fig.\ref{fig1}, see also \cite{Levin2006topological}. That is because, as shown in the previous section, the number of boundaries $n_{\partial}(\Lambda)$ gives the number of topological corrections $\gamma$ to the purity $P_{\Lambda}$ of the related subsystem $\Lambda$. Specifically, consider Fig.\ref{fig1} $(b)$ first: we have $n_{\partial}(AB)=n_{\partial}(BC)=n_{\partial}(B)=n_{\partial}(ABC)=1$, and thus according to Eq.\eqref{toprule} we have $2\gamma -2\gamma=0$, while for Fig.\ref{fig1} $(a)$: $n_{\partial}(AB)=n_{\partial}(BC)=1$ and $n_{\partial}(B)=n_{\partial}(ABC)=2$, thus $2\gamma-4\gamma=-2\gamma$.

In \cite{Kitaev2006topological}, it was argued, through TQFT arguments, that a local deformation of the boundary will not affect the topological entropy. For the lattice model, we can prove the following lemma that affirms the same property:
\begin{lemma}\label{lemma}
The TP for the state $\Psi_{0}\in\mathcal{L}$ is stable under local boundary deformations. 
\end{lemma}
\begin{proof} Consider a subset $X\subset V$ and consider the boundary transformation $ABC\rightarrow ABC\cup X$. Let us suppose, without loss of generality, that $X\cap \partial A\neq \emptyset$, cfr. Fig.\ref{fig1} $(c)$. Thus, according to  Eq.\eqref{purity}, we have:
\ba
P_{AB}(\Psi_0)&\rightarrow& P_{AB\cup X}=2^{-|\partial(AB\cup X)|+\Gamma_{AB\cup X}}\\ P_{ABC}(\Psi_0)&\rightarrow& P_{ABC\cup X}=2^{-|\partial(ABC\cup X)|+\Gamma_{AB\cup X}}
\ea
while $P_{B}(\Psi_0)$ and $P_{BC}(\Psi_0)$ are unmodified. Then note that $| \partial(AB\cup X)|+|\partial (BC)|=|\partial B|+|\partial(ABC\cup X)|$ and that $\Gamma_{AB\cup X}+\Gamma_{BC}-\Gamma_{B}-\Gamma_{ABC\cup X}=-2\gamma$ because $\gamma_{AB\cup X}+\gamma_{BC}=\gamma_{B}+\gamma_{ABC\cup X}$.
\end{proof}
The request of the boundary modification being local is extremely important. Indeed, it is always true that, as sets, $\gamma_{AB}+\gamma_{BC}=\gamma_{B}+\gamma_{ABC}$. However, only if the boundary modifications are local then the number of disconnected boundaries $n_\partial$ is unchanged.

The above lemma is necessary for the TP to be a good detector of the TO. In particular, if $P_{top}(\sigma)=1$, then $\sigma$ is a trivial topological state, while if $P_{top}(\sigma)<1$, $\sigma$ is expected to be in a non trivial topological phase.

%%%%%%%%%%%%%%%%%%%%%%%%%%%%%%%%%%%%%%%%%%%%%%%%%%%%%%%%%%%%%%%%%%%%%%%%%%%%%%%%%%%%%%%%%%%%%%%%%%%%%%%%%%%%%%%%%%%%%%%%%%%
\section{Stability of topological purity}
In this section, we establish a noise model based on quenched disorder, and show how the topological order behaves under the noise model. We consider a random local unitary noise on the spins in the lattice. Under this noise model, one can compute the average purity in a subsystem due to this noise. This is the purity one would measure in an experiment if the measurement time-scales are much longer than the random fluctuations in the unitary noise. Then one can define a topological purity associated with these average purities. If the noise is zero, the average is just the purity of the initial state, and the two topological purities coincide. 

We want to show that, if the random unitaries build up a shallow quantum circuit, then the topological purity is constant. The proof is quite technical, so let us first give an intuitive explanation. We first show that, the purity in the subsystem $\Lambda$ averaged over all the transformations of the state under the noise model is equal to the purity in the {\em initial state} for a subsystem with a {\em different boundary}, depending only on which spins were affected by the random quantum noise. It is a very non trivial fact that under such kind of noise, the average purity results in just a boundary transformation for the purity in the initial state. Then it follows that, if the noise is shallow, only shallow deformations of the boundary are possible, and by Lemma 1 the topological purity is preserved.
\subsection{Topological purity and phases}\label{phases}
Let us now dive into the technical details of the noise model and the robustness of topological purity.
Let $\mathcal H_V \simeq \mathbb C^D\simeq \mathbb C^{d\otimes N}$ be the $D-$dimensional Hilbert space of $N$ qudits in a set $V$. Here, the Hilbert space of the $x-$th qudit is denoted by $ \mathcal H_x \simeq \mathbb C^d$. Let $\Lambda\subset V$ be a subset of these qudits and $\mathcal H_\Lambda= \otimes_{x\in\Lambda} \mathcal H_x$ the corresponding Hilbert space. 
Let $\tilde{T}_\Lambda$ be the order two permutation (swap) operator on $ \mathcal H_\Lambda^{\otimes 2}$ and let ${T}_\Lambda =\tilde{T}_\Lambda\otimes \bbbone_{\bar{\Lambda}}$ be its trivial completion on the full $\mathcal H_{\Lambda}^{\otimes 2}\otimes \mathcal H_{\bar{\Lambda}}^{\otimes 2}$.

The purity of the state $\sigma$ in the bipartition $\mathcal{H}_{\Lambda}\otimes \mathcal{H}_{\bar \Lambda}$  is given by
\ba
P_{\Lambda}(\sigma)\equiv \tr_\Lambda \sigma_\Lambda^2 =\tr(\sigma^{\otimes 2} T_\Lambda)\equiv\aver{T_{\Lambda}}_{\sigma^{\otimes 2}}
\label{purityasexpvalue}
\ea
where $\sigma_{\Lambda}:=\tr_{\bar\Lambda}\sigma$. The above chain of relations is telling us that the purity is from both the analytical point of view and the experimental point of view a quantity defined on two copies of the Hilbert space $\mathcal{H}$. In practice, in order to experimentally measure the purity of a quantum state $\sigma$ in a given bipartition $\mathcal{H}_{\Lambda}\otimes \mathcal{H}_{\bar\Lambda}$, one needs three steps: $(i)$ to prepare two identical copies of $\sigma$, $(ii)$ to build the observable \textit{swap operator} on the subspace $\Lambda$ and finally, $(iii)$ to take the quantum expectation value of $T_{\Lambda}$ in $\sigma\otimes \sigma$, namely $\aver{T_{\Lambda}}_{\sigma^{\otimes 2}}$. Similarly, in order to experimentally measure the topological purity, defined in Eq.\eqref{Eq:toppurity}, one needs to repeat the steps $(i)$, $(ii)$ and $(iii)$ for the four observables $T_{AB}, T_{BC}, T_{B}, T_{ABC}$ and then combine them in the following way:
\be\label{sasa2}
P_{top}(\sigma)=\frac{\aver{T_{AB}}_{\sigma^{\otimes 2}}\aver{T_{BC}}_{\sigma^{\otimes 2}}}{\aver{T_{B}}_{\sigma^{\otimes 2}}\aver{T_{ABC}}_{\sigma^{\otimes 2}}}
\ee

Since the purity is defined on $\mathcal{H}^{\otimes 2}$, we define a noise model on states living on two copies of the Hilbert space in the following way:  let $X\subset V$ be a set of qudits with Hilbert space $\mathcal{H}_{X}:=\otimes_{x\in X} \mathcal H_x$ and $d_{X}:=\dim \mathcal{H}_{X}$. Let $U_X$ be a local unitary operator operating on the region $X$, i.e. operating on all the qubits contained in $X$. Let $U_{X}^{\otimes 2}$ be two copies of $U_{X}$, then after operating on $\sigma^{\otimes 2}$ with the unitary $U_{X}^{\otimes 2}$, we have $\sigma^{\otimes 2} \mapsto U_{X}^{\otimes 2}\sigma^{\otimes 2} U_{X}^{\dag\otimes 2}$ and the purity becomes
\be
\aver{T_{\Lambda}}_{\sigma^{\otimes 2}}\mapsto \aver{T_{\Lambda}}_{U_{X}^{\otimes 2}\sigma^{\otimes 2} U_{X}^{\otimes 2}}\equiv\tr{\left(T_{\Lambda}U_{X}^{\otimes 2}\sigma^{\otimes 2} U_{X}^{\otimes 2}\right)}
\ee

We now choose $U_{X}$ to be a random unitary operator and define the following quantum map acting on $\sigma^{\otimes 2}$:
\be\label{rxdef}
R_{X}(\sigma^{\otimes 2}):=\int \de \mu(U|X) (U_X)^{\otimes 2} \sigma^{\otimes 2} (U_X)^{\dag\otimes 2}
\ee
where $\de \mu(U|X)$ is the Haar measure over the unitary group $\mathcal{U}(\mathcal{H}_X)$. Therefore, fixed $X\subset V$, the map $R_{X}(\cdot)$ randomizes over the action of the full unitary group on $\mathcal{H}_{X}^{\otimes 2}$. Thus, after the noise on $X$, the purity becomes:
\be
\aver{T_{\Lambda}}_{\sigma^{\otimes 2}}\mapsto \aver{T_{\Lambda}}_{R_X(\sigma^{\otimes 2})}\equiv \tr{\left(T_{\Lambda}R_X(\sigma^{\otimes 2})\right)}
\ee
note that the above operation is no more acting independently on the single copies of $\mathcal{H}$, but it is entangling them in $\mathcal{H}^{\otimes 2}$. So far this is a single $X$ noise model. In order to generalize it to more than one single $X$ domain, consider an ordered string of subsets $S=\{\tilde{X}_{1},\dots, \tilde{X}_{k}\}$ and random unitary operators $U_{\tilde{X}_{i}}^{\otimes 2}, i=1,\dots, k$ operating on the corresponding subset $\tilde{X}_{i}$ and acting on $\sigma^{\otimes 2}$ in an ordered way, namely $\sigma^{\otimes 2}\mapsto U_{\tilde{X}_k}^{\otimes 2}\cdots U_{\tilde{X}_1}^{\otimes 2}\sigma^{\otimes 2}U_{\tilde{X}_1}^{\dag\otimes 2}\cdots U_{\tilde{X}_k}^{\dag\otimes 2}$. We define the quantum map randomizing over the action of these gates as:
\be\label{rsdef}
\mathcal{R}_{S}(\sigma^{\otimes 2}):=R_{\tilde{X}_k}\cdots R_{\tilde{X}_1}(\sigma^{\otimes 2})
\ee
where 
\be
R_{\tilde{X}_i}\,: \mathcal{O}\mapsto R_{X_i}(\mathcal{O}):=\int \de \mu(U|\tilde{X}_i) (U_{\tilde{X}_i})^{\otimes 2} \mathcal{O} (U_{\tilde{X}_i})^{\dag\otimes 2}, \quad \mathcal{O}\in\mathcal{B}(\mathcal{H}^{\otimes 2})
\ee
For each string of domains $S$, the action of $\mathcal{R}_{S}$ on a state of $\mathcal{H}^{\otimes 2}$ describes the average action of a given random quantum circuit operating in the region $\tilde{X}_i \in S$, therefore at this point we define the set $\mathcal{S}$ of all such strings:
\be
\mathcal{S}:=\{S=\{\tilde{X}_{1},\dots, \tilde{X}_{k}|\, \tilde{X}_{i}\subset V, \,i=1,\dots, k\}, \, k\in\mathbb{N}\}
\label{defmathcalS}
\ee
It is straightforward to see that, for any subset $\tilde{\mathcal{S}}\subset \mathcal{S}$, the action of $\mathcal{R}_{S}$ on a state $\sigma^{\otimes 2}$ varying $S\in\tilde{\mathcal{S}}$  creates an ensemble of states living on $\mathcal{H}^{\otimes 2}$ as follows:
\be
\mathcal{E}_{\tilde{\mathcal{S}}}(\sigma^{\otimes 2}):=\{\mathcal{R}_{S}(\sigma^{\otimes 2})\in\mathcal{B}(\mathcal{H}^{\otimes 2})\, |\,S\in\tilde{\mathcal{S}}\}
\ee
i.e. the ensemble of states $\mathcal{E}_{\tilde{\mathcal{S}}}(\sigma^{\otimes 2})$ contains all the states $\Psi_S$ living in $\mathcal{H}^{\otimes 2}$ obtained by the action of $\mathcal{R}_{S}$ varying $S$ in a subset $\tilde{\mathcal{S}}$ of $\mathcal{S}$, defined in Eq.\eqref{defmathcalS}. Notice that each string of ordered domains $S$ describes a quantum circuit consisting of random gates with support on $\tilde{X}_i\in S$. 

Now we are ready to state the main result of this paper. Consider $\ket{\Psi_0}\in\mathcal{L}$ the ground state of the quantum double model (cfr. Sec.\ref{sec:QuantumDouble}) and a string of domains $S\in\mathcal{S}$. Then define $\Psi_{S}\in\mathcal{B}(\mathcal{H}^{\otimes 2})$ as the quantum state, living on two copies of $\mathcal{H}$ and obtained by the action of $\mathcal{R}_{S}(\cdot)$, $\Psi_{S}:=\mathcal{R}_{S}(\Psi_{0}^{\otimes 2})$. To characterize the topological nature of $\Psi_{S}$ we lay down the following definition:
\be
\tilde{P}_{top}(\Psi_{S}):=\frac{\aver{T_{AB}}_{\Psi_S}\aver{T_{BC}}_{\Psi_S}}{\aver{T_{B}}_{\Psi_S}\aver{T_{ABC}}_{\Psi_S}}
\label{toppurityforPsiS}
\ee
in the same fashion of the topological purity defined in Eq.\eqref{Eq:toppurity}, i.e. the ratio of the expectation values of the swap operators $T_{AB}, T_{BC}, T_{B}, T_{ABC}$. Notably the above definition is the extension of the topological purity in Eq.\eqref{sasa2} to non-product states of $\mathcal{H}^{\otimes 2}$. 

For product states $\Psi^{\otimes 2}$ of $\mathcal{H}^{\otimes 2}$ we indeed have: $\tilde{P}_{top}(\Psi^{\otimes 2})=P_{top}(\Psi)$. 
Since the TP is a quantity that naturally extends its definition for states of $\mathcal{H}^{\otimes 2}$, we refer to $\tilde{P}_{top}(\Psi_{S})$ as the \textit{topological purity} of $\Psi_{S}$ for non-product state as well.  
This protocol to detect topological order under the noisy channel we defined is experimentally feasible using the techniques in \cite{satzinger2021realizing}.

Now we can show the main result of this paper: in the following theorem we prove that the topological purity attains a constant value in the ensembles of states obtained from both the ground state of the toric code and a topologically trivial pure state, provided that the subset $\tilde{\mathcal{S}}\subset \mathcal{S}$ contains strings of domains $S$ describing shallow quantum circuits. Since this definition contains circuits with trivial action, this value is also the value of the topological purity in the initial state.
\begin{theorem*}
Let $\Psi_0$ be the ground state of a quantum double model and let $\Phi$ be a pure, topologically trivial quantum state. Let $\tilde{\mathcal S}\subset \mathcal{S}$ be the set defined in Definition \ref{definS}, then the topological purity is constant in the following ensembles of states: $\mathcal{E}_{\tilde{\mathcal{S}}}(\Psi_{0}^{\otimes 2})$, $\mathcal{E}_{\tilde{\mathcal{S}}}(\Phi^{\otimes 2})$, namely:
\be
\tilde{P}_{top}(\Psi_{S})=2^{-2\gamma}, \quad \forall \Psi_{S}\in \mathcal{E}_{\tilde{\mathcal{S}}}(\Psi_{0}^{\otimes 2}),
\ee
and 
\be
\tilde{P}_{top}(\Phi)=1, \quad \forall \Phi\in \mathcal{E}_{\tilde{\mathcal{S}}}(\Phi^{\otimes 2})
\ee
\end{theorem*}

In the above theorem, $\tilde{S}\subset S$ is a subset of $\mathcal{S}$ which will be rigorously defined in Sec.\ref{mainresult}; morally $\tilde{S}$ contains strings $S$ of domains describing a shallow random quantum circuit that do not destroy the topological nature of $\Psi_{0}^{\otimes 2}$.
Since we found that the TP gets a constant value in two distinct ensemble of states, we claim that the topological purity is stable in the topological ordered phase ${\mathcal{E}}_{\tilde{S}}(\Psi_{0}^{\otimes 2})$ and in the topological trivial phase ${\mathcal{E}}_{\tilde{S}}(\Phi^{\otimes 2})$.

The proof of this theorem is in Sec.\ref{mainresult}. As we stated at the beginning of this section, the proof descends from the fact that the purity averaged over the noise is equal to the purity in the initial state, but with a different boundary. In the next section, we explain this very non trivial result, after which we will be ready to show the proof of the theorem.

\subsection{Purity dynamics under random quantum circuits}\label{model}
In this section, we show how, under the noise model defined by the quantum map Eq.(), the evolution of the purity becomes a boundary evolution for the purity in the initial state.

 Consider a state $\sigma$ and the swap operator $T_{\Lambda}$ defined in the region $\Lambda$. As shown in the previous section the purity of $\sigma$ in the region $\Lambda$ is the expectation value of the swap operator $T_{\Lambda}$ computed on two copies of the state $\sigma$, $P_{\Lambda}(\sigma)\equiv\aver{T_{\Lambda}}_{\sigma^{\otimes 2}}$. Let $R_{X}$ be the quantum map defined in Eq.\eqref{rxdef} and consider the expectation value of $T_{\Lambda}$ in the state $R_{X}(\sigma^{\otimes 2})$; because $R_{X}(\cdot)$ is an hermitian and self-dual operator\cite{Hammalungo_2012}, we can equivalently write:
\be
\aver{T_{\Lambda}}_{R_{X}(\sigma^{\otimes 2})}=\aver{R_{X}(T_{\Lambda})}_{\sigma^{\otimes 2}}
\label{heisenberg}
\ee
i.e. the expectation value of the swap operator $T_{\Lambda}$ on the state $R_{X}(\sigma^{\otimes 2})$ is equal to the expectation value of the evolved swap operator $R_{X}(T_{\Lambda})$, i.e. the image of $T_{\Lambda}$ under the map $R_{X}(\cdot)$, on the original state $\sigma^{\otimes 2}$. In practice, we are considering the Heisenberg picture for the evolution of the swap operator. This point of view is convenient for us, because - thanks to the simple equation \eqref{heisenberg} -  the purity dynamics can be described as the dynamics of the boundary $\partial \Lambda$ of the region $\Lambda$. First of all, let $X\subset V$ be a domain and let us compute the action of the map $R_{X}(\cdot)$ on the swap $T_{\Lambda}$. One can show\cite{Hammalungo_2012}:
\be
R_X(T_\Lambda)=\begin{cases}
N_{d_{\Lambda\setminus X}} T_{\Lambda\setminus X}+N_{d_{\Lambda\cup X}}T_{\Lambda\cup X},& X\cap\partial\Lambda \neq \emptyset , \\
T_\Lambda, &  X\cap\partial\Lambda=\emptyset
\end{cases}
\label{algebra}
\ee
where $N_{d_{\Lambda\setminus X}}:=(d_{ X}^{2}-d_{\Lambda\cap X}^{2})d_{\Lambda\cap X}^{-1}/(d_{X}^{2}-1)$ and $N_{d_{\Lambda\cup X}}:=d_{X}(d_{\Lambda\cap X}^{2}-1)d_{\Lambda\cap X}^{-1}/(d_{X}^{2}-1)$ and $d_{\Lambda\cap X}=\dim \mathcal{H}_{\Lambda\cap X}$.
\begin{figure}[h!]
    \centering
    \resizebox{9cm}{!}{\includegraphics[scale=0.6]{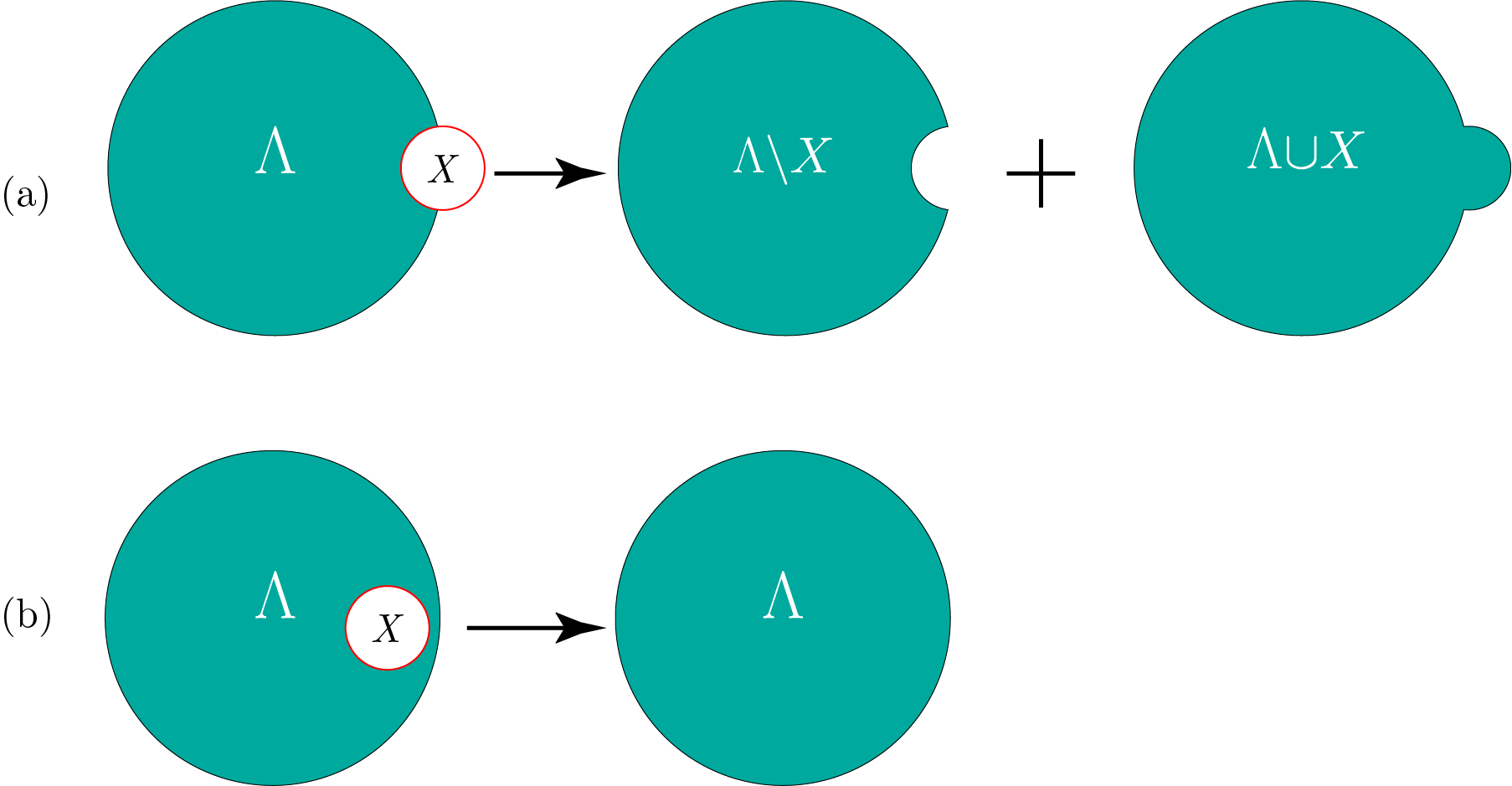}}
    \caption{An illustration of the action of the superoperator $R_{X}$ on the swap operator $T_{\Lambda}$ with support on the region $\Lambda\subset V$. In $(a)$ the domain $X$ has a non trivial overlap with the boundary $\partial\Lambda$ and, according to Eq.\eqref{algebra}, its action gives a linear combination of two domains, namely $\Lambda\setminus X$ and $\Lambda \cup X$. In $(b)$ the domain $X\subset \Lambda$ is completely contained in $\Lambda$ and its action is trivial. Note that we neglected the prefactors $N_{d_{\Lambda\setminus X}}$ and $N_{d_{\Lambda\cup X}}$, cfr. Eq.\eqref{algebra}.}
    \label{composition}
\end{figure}

A simple representation of this action is provided in Fig.\ref{composition}. Note that we can compactly write the above action as follows:
\be
R_{X}(T_{\Lambda})=(1-f(X,\Lambda))T_{\Lambda}+f(X,\Lambda)[N_{d_{\Lambda\cup X}}T_{\Lambda\cup X}+N_{d_{\Lambda\setminus X}}T_{\Lambda\setminus X}]
\label{compactalgebra}
\ee
where $f\,:\, (X,\Lambda)\rightarrow \mathbb{R}$:
\be
f(X,\Lambda)=\begin{cases}1, \quad X\cap \partial\Lambda \neq \emptyset\\
0, \quad X\cap \partial\Lambda = \emptyset
\end{cases}
\label{deff}
\ee

Eqs. \eqref{algebra} and \eqref{compactalgebra} are telling us that if $X$ intersects the boundary of $\Lambda$, the expectation value of $T_\Lambda$ on $R_{X}(\sigma^{\otimes 2})$ becomes the linear combination of the expectation values on the original state $\sigma^{\otimes 2}$ of swap operators with different boundaries, namely  $\Lambda \cup X$ and $\Lambda \cap X$. If $X$ is completely inside or outside $\Lambda$, the expectation value is unchanged. Thus, thanks to the duality in Eq.\eqref{heisenberg} we can write the expectation value of $R_{X}(T_{\Lambda})$ on $\sigma^{\otimes 2}$ as a linear combination of the purity of $\sigma$ in different domains:
\be
\aver{R_{X}(T_{\Lambda})}_{\sigma^{\otimes 2}}=(1-f(X,\Lambda))P_{\Lambda}(\sigma)+f(X,\Lambda)[N_{d_{\Lambda\cup X}}P_{\Lambda\cup X}(\sigma)+N_{d_{\Lambda\setminus X}}P_{\Lambda\setminus X}(\sigma)]
\ee
where $P_{\Lambda}(\sigma)=\aver{T_{\Lambda}}_{\sigma^{\otimes 2}}$ etc. Now, in order to generalize the above discussion to more than one domain $X$, consider the action of the quantum map $\mathcal{R}_{S}(\cdot)$ defined in Eq.\eqref{rsdef}. Although $\mathcal{R}_{S}(\cdot)$ is no more hermitian, we can exploit the duality as well and write:
\be
\aver{T_{\Lambda}}_{\mathcal{R}_{S}(\sigma^{\otimes 2})}=\aver{\mathcal{R}_{S}^{\dag}(T_{\Lambda})}_{\sigma^{\otimes 2}}
\ee
where $\mathcal{R}_{S}^{\dag}(\cdot)=R_{\tilde{X}_1}\cdots R_{\tilde{X}_k}(\cdot)$. As one can see the adjoint operator $\mathcal{R}^{\dag}_{S}(\cdot)$ is always the same operator $\mathcal{R}(\cdot)$ with a different ordering of the domains of $S$. Thus, defining the ordered subset  
\be \label{bars}
\bar{S}=\{\tilde{X}_{k},\dots, \tilde{X}_{1}\, |\tilde{X}_{i}\in S\}
\ee
we can write $\mathcal{R}^{\dag}_{S}(\cdot)=\mathcal{R}_{\bar{S}}(\cdot)$. Here \textit{ordered} means that, given $\tilde{X}_{i},\tilde{X}_{j}\in \bar{S}$ with $i>j$, the map $R_{\tilde{X}_j}(\cdot)$ acts after the map $R_{\tilde{X}_i}(\cdot)$. In order to avoid confusion, let us re-define the subsets as $X_{j}=\tilde{X}_{k+1-j}$, so that $\bar{S}=\{X_{1},\dots, X_{k}\,|\, X_{j}=\tilde{X}_{k+1-j}, \,\tilde{X}_{k+1-j}\in S\}$ and $\mathcal{R}_{\bar{S}}(T_{\Lambda})=R_{X_k}\cdots R_{X_1}(\cdot)$. By duality, the expectation value of $\mathcal{R}_{S}(T_{\Lambda})$ is always a linear combination of purities of $\sigma$:
\be
\aver{T_{\Lambda}}_{\mathcal{R}_{S}(\sigma^{\otimes 2})}=\aver{\mathcal{R}_{\bar{S}}(T_{\Lambda})}_{\sigma^{\otimes 2}}=\sum_{\Lambda_{\alpha}\in\mathcal{Y}^{(k)}(\Lambda)}m_{\Lambda_{\alpha}}P_{\Lambda_{\alpha}}(\sigma)
\label{compact}
\ee
where we defined the set of domains $\mathcal{Y}^{(k)}(\Lambda):=\{\Lambda,\Lambda\cup X_{1},\Lambda\cup X_{2},\dots, \Lambda\cup X_{1}\setminus X_{2},\dots\}$; in the r.h.s of \eqref{compact} the coefficients $m_{\Lambda_{\alpha}}$ depend on the particular choice of the ordered string $S$. In order to make the notation clearer, let us write the expression for $k=2$ explicitly:
\ba
\aver{\mathcal{R}_{\bar{S}}(T_{\Lambda})}_{\sigma^{\otimes 2}}&\equiv&\aver{R_{X_2}R_{X_1}(T_{\Lambda})}_{\sigma^{\otimes 2}}=(1-f(X_1,\Lambda))(1-f(X_2,\Lambda))P_{\Lambda}(\sigma)\nonumber\\&+&(1-f(X_1,\Lambda))f(X_{2},\Lambda)[N_{d_{\Lambda\cup X_2}}P_{\Lambda\cup X_{2}}(\sigma)+N_{d_{\Lambda\setminus X_2}}P_{\Lambda\cup X_{2}}(\sigma)]\nonumber\\
&+&f(X_{1},\Lambda)(1-f(X_{2},\Lambda\cup X_1))N_{d_{\Lambda\cup X_1}}P_{\Lambda\cup X_1}(\sigma)\nonumber\\&+&
f(X_{1},\Lambda)(1-f(X_{2},\Lambda\setminus X_1))N_{d_{\Lambda\setminus X_1}}P_{\Lambda\setminus X_1}(\sigma)\\&+&
f(X_{1},\Lambda)f(X_{2},\Lambda\cup X_{2})(N_{d_{\Lambda\cup X_1}}N_{d_{\Lambda\cup X_2}}P_{\Lambda\cup X_1\cup X_2}(\sigma)+N_{d_{\Lambda\cup X_1}}N_{d_{\Lambda\setminus X_2}}P_{\Lambda\cup X_1\setminus X_2}(\sigma))\nonumber\\&+&
f(X_{1},\Lambda)f(X_{2},\Lambda\setminus X_{2})(N_{d_{\Lambda\setminus X_1}}N_{d_{\Lambda\cup X_2}}P_{\Lambda\setminus X_1\cup X_2}(\sigma)+N_{d_{\Lambda\setminus X_1}}N_{d_{\Lambda\setminus X_2}}P_{\Lambda\setminus X_1\setminus X_2}(\sigma))\nonumber
\label{examplef}
\ea
where $\bar{S}=\{X_1,X_2\}$ and $m_{\Lambda}\equiv(1-f(X_1,\Lambda))(1-f(X_2,\Lambda))$, $m_{\Lambda\setminus X_1\setminus X_2}\equiv N_{d_{\Lambda\setminus X_1}}N_{d_{\Lambda\setminus X_2}}f(X_{1},\Lambda)f(X_{2},\Lambda\setminus X_{2})$, etc. 
The model is completely general: once one has chosen the string $S$ and the domain $\Lambda$ the functions $f$ are is either $0$ or $1$ according to the definition in Eq.\eqref{deff}. It is worth noting that the ordering of the domains $X_{i}\in \bar{S}$ is very important; consider $X_{1},X_{2}\in \bar{S}$ and note it can be the case that $X_{2}\cap \partial(\Lambda/X_{1})\neq 0$ while $X_{2}\cap \partial\Lambda=0$ and so in the if $X_{2}$ acts before $X_{1}$ it does not have any effect, see Fig.\ref{ORDERING} for a pictorial proof.
\begin{figure}[H]
    \centering
    \resizebox{9cm}{!}{\includegraphics[scale=0.6]{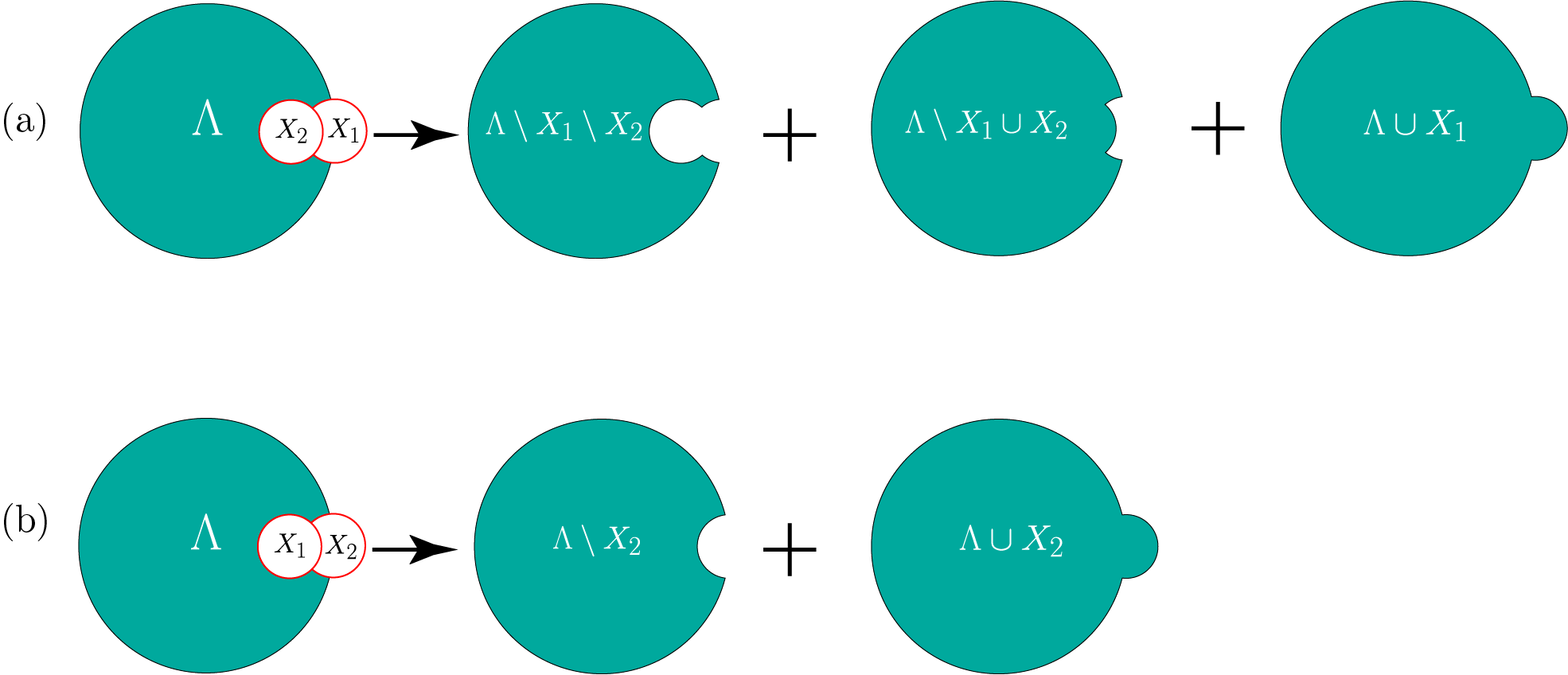}}
    \caption{The figure shows how different orderings of the same domains $X_{i}\in \bar{S}$ can give different results. In $(a)$, $\bar{S}_{(a)}=\{X_{1},X_{2}\}$ and thus we have $f(X_1,\Lambda)=1$, $f(X_{2},\Lambda\setminus X_{1})=1$, $f(X_{2},\Lambda\cup X_{1})=f(\Lambda,X_{2})=0$; therefore $\aver{R_{\bar{S}_{(a)}}(T_{\Lambda}}_{\sigma^{\otimes 2}}=N_{\Lambda\setminus X_{1}}P_{\Lambda \cup X_{1}}(\sigma)+N_{\Lambda\setminus X_{1}}N_{\Lambda\setminus X_{1}\setminus X_{2}}P_{\Lambda\setminus X_{1}\setminus X_{2}}(\sigma)+N_{\Lambda\setminus X_{1}}N_{\Lambda\setminus X_{1}\cup X_{2}}P_{\Lambda\setminus X_{1}\cup X_{2}}(\sigma)$, cfr. Eqs. \eqref{deff} and \eqref{examplef}. In $(b)$, $\bar{S}_{(b)}=\{X_{2},X_{1}\}$ and thus we have $f(X_{1},\Lambda)=0$, $f(X_{2},\Lambda)=1$; therefore $\aver{R_{\bar{S}_{(b)}}(T_{\Lambda}}_{\sigma^{\otimes 2}}=N_{\Lambda\setminus X_{2}}P_{\Lambda\setminus X_{2}}(\sigma)+N_{\Lambda \cup X_{2}}P_{\Lambda \cup X_{2}}(\sigma)$.}
    \label{ORDERING}
\end{figure}
%%%%%%%%%%%%%%%%%%%%%%%%%%%%%%%%%%%%%%%%%%%%%%%%%%%%%%

%%%%%%%%%%%%%%%%%%%%%%%%%%%%%%%%%%%%%%%%%%%%%%%%%%%%
% FIGURE
%%%%%%%%%%%%%%%%%%%%%%%%%%%
%%%%%%%%%%%%%%%%%%%%%%%%%%%

%%%%%%%%%%%%%%%%%%%%%%%%%%%%%%%%%%%%%%%%%%%%%%%%%%%%%%

\subsection{Proof of the main result}\label{mainresult}
This section is devoted to the proof of the main result of the paper. In virtue of the above section, we can re-write Eq.\eqref{toppurityforPsiS} for the topological purity of the state $\Psi_{S}=\mathcal{R}_{S}(\Psi_{0}^{\otimes 2})$ in terms of expectation values of evolved swap operators $\mathcal{R}_{S}(T_{\Lambda})$ for $\Lambda$ being $AB, BC, B$ and  $ABC$, namely:
\be
\tilde{P}_{top}(\Psi_{S})=\frac{\aver{\mathcal{R}_{\bar{S}}(T_{AB})}_{\Psi_{0}^{\otimes 2}}\aver{\mathcal{R}_{\bar{S}}(T_{BC})}_{\Psi_{0}^{\otimes 2}}}{\aver{\mathcal{R}_{\bar{S}}(T_{B})}_{\Psi_{0}^{\otimes 2}}\aver{\mathcal{R}_{\bar{S}}(T_{ABC})}_{\Psi_{0}^{\otimes 2}}}
\label{toppurityduality}
\ee
Before moving on to the formal proof, it is useful to understand how the topological purity $\tilde{P}_{top}(\Psi_{S})$ changes as a function of $S\in \mathcal{S}$. To do this, it is useful to recall that, the result of the action of $\mathcal{R}_{\bar S}(\cdot)$ on the swap operator $T_{\Lambda}$ is a linear combination of swap operators defined on different domains $\Lambda_{\alpha}$ obtained from $\Lambda$ by local boundary modifications. Since the topological purity $\tilde{P}_{top}(\Psi_{S})$ is given by the product (and ratio) of expectation values of four swap operators $T_{AB}, T_{BC}, T_{B}, T_{ABC}$ (cfr. Eq.\eqref{toppurityduality}), the boundary modifications can cancel each other and thus not every boundary modification changes the topological purity, suggesting instead the existence of strings $S$ able to preserve it. Because of Lemma \ref{lemma}, boundary modifications that do not change the topology of the region $ABC$ cannot change the topological purity. At the same time, some topologically non-trivial modifications can preserve the TP as well. Consider for example Fig.\ref{figsasa} $(a)$ and $(b)$: those boundary modifications, although changing the topology of the region $ABC$ return the same TP. As one expects, among those modifications changing the topology of $ABC$, there are ones breaking the topological purity, e.g. see Fig.\ref{figsasa} $(c)$ and $(d)$. Therefore, our main concern is to identify which topologically non trivial boundary modifications preserve the TP and which ones break it. One can observe by looking at Fig.\ref{figsasa} that: boundary modifications able to return a different topological purity are of two kinds: the ones that cut apart the donut shape of the region $ABC$ connecting the two disconnected $D$ regions Fig.\ref{figsasa} $(c)$ and the ones that connect two non-contiguous regions, e.g. $B_{left}$ and $B_{right}$ Fig.\ref{figsasa} $(d)$, because for these cases only we have that $\Gamma_{A^{\prime}B^{\prime}}+\Gamma_{B^{\prime}C^{\prime}}-\Gamma_{B^{\prime}}-\Gamma_{A^{\prime}B^{\prime}C^{\prime}}\neq \Gamma_{AB}+\Gamma_{BC}-\Gamma_{B}-\Gamma_{ABC}$, where $A^{\prime}, B^{\prime}$ and $C^{\prime}$ are the modified regions.
\begin{figure}[H]
    \centering
    \resizebox{12cm}{!}{\includegraphics[scale=0.6]{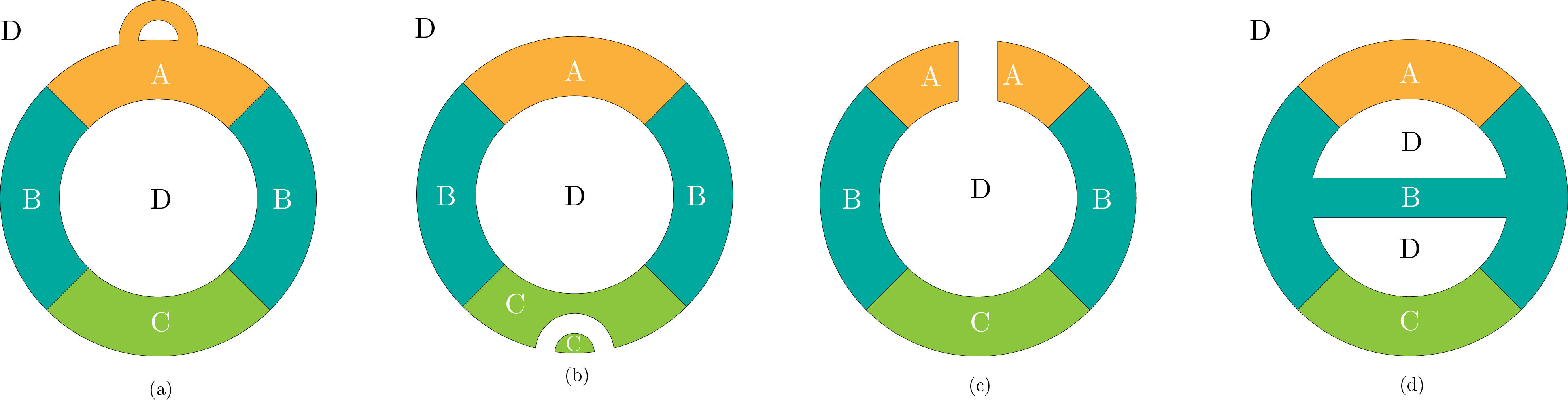}}
    \caption{An illustration of how some modifications of the boundary given in Fig.\ref{fig1} and used for the computation of the TP can either affect the TP or not. In particular, $(a)$ and $(b)$ show the same topological purity of Fig.\ref{fig1}$(a)$, while $(c)$ and $(d)$ do not. $(a)$: we have $n_{\partial}(AB)=n_{\partial}(B)=2$, $n_{\partial}(ABC)=3$, $n_{\partial}(BC)=1$ and thus $P_{top}=2^{-2\gamma}$. $(b)$: we have $n_{\partial}(AB)=1$, $n_{\partial}(BC)=n_{\partial}(B)=2$, $n_{\partial}(ABC)=3$ and thus $P_{top}=2^{-2\gamma}$. $(c)$: we have $n_{\partial}(AB)=n_{\partial}(B)=2$ and $n_{\partial}(BC)=n_{\partial}(ABC)=1$ and thus $P_{top}=1$; $(d)$: we have $n_{\partial}(AB)=n_{\partial}(BC)=2$, $n_{\partial}(B)=1$, $n_{\partial}(ABC)=3$ and thus $P_{top}=1$.}
    \label{figsasa}
\end{figure}

Equipped with some intuitions of what is going on, let us now consider a string $S$. Because of the importance of the ordering of the subsets(cfr. Sec.\ref{model}) we can focus on \textit{ordered and contiguous} substrings $S\supset S^{\prime}=\{Y_{1},\ldots Y_{|S^{\prime}|}\,|\, Y_{i}\in S\}$. Ordered and contiguous means: $Y_{i-1}\cap \partial Y_{i}\neq \emptyset$ and $Y_{i}\cap \partial Y_{i+1}\neq \emptyset\, \forall i.$
According to the rules described in the above section, the join of the domains $Y_{i}\in S^{\prime}$, $\mathcal{Y}\equiv\bigcup_{i\in S^{\prime}}Y_{i}$, can result in either (i) cancel part of the region $ABC$ and cut apart its donut shape, as $\Lambda\setminus \mathcal{Y}$ or (ii) connect two non-contiguous parts of $ABC$, as $ABC\cup\mathcal{Y}$. Of course can exist many substrings of $S$ that, joining together, result $(i)$ or $(ii)$. The tale we presented above is the core of the proof of the main theorem; conditions $(i)$ and $(ii)$ are sketched in Figs. \ref{conditioni}, \ref{conditionii} and \ref{sasa} and summarized more rigorously in Definition \ref{definS}.
\begin{defin}\label{definS}
$\mathcal{S}\supset\tilde{\mathcal{S}}=\{S=\{\tilde{X}_1,\dots, \tilde{X}_{k}\}\}$ such that $\forall S\in\mathcal{S}$ the following is not satisfied: 
\begin{enumerate}[(I)]
    \item given $\bar{S}=\{X_{1},\dots, X_{k}\,|\, \tilde{X}_{i}\in S,\,            Z_{j}=\tilde{X}_{k+1-j}\}$ defined in Eq.\eqref{bars}, there exist two ordered subsets $ \bar{S} \supset     S^{\prime}=\{Y_{1},\dots, Y_{|S^{\prime}|}\,|\, Y_{i}\in \bar{S}\}$ and $\bar{S} \supset S^{\prime\prime}=\{Z_{1},\dots, Z_{|S^{\prime\prime}|}\,|\, Z_{i}\in \bar{S}\}$ where $Y_{i-1}\cap \partial Y_{i}, X_{i}\cap \partial Y_{i+1}\neq \emptyset$, $Z_{i-1}\cap \partial Z_{i}, Z_{i}\cap \partial Z_{i+1}\neq \emptyset$ $\forall i$ and $Y_{|S^{\prime}|}\cap Z_{|S^{\prime\prime}|}\neq \emptyset$ that fulfill one of the following conditions (see Figs.\ref{conditioni} and \ref{conditionii}  for a pictorial overview):
\begin{enumerate}[(i)]
    \item  $Y_{1}\cap \partial(ABC)_{1}\neq \emptyset$, $Z_{1}\cap \partial(ABC)_{2}\neq \emptyset$,
    \item $ Y_{1}\cap \partial(\Lambda_{\alpha}) \neq \emptyset$, $Z_{1}\cap \partial(\Lambda_{\beta})\neq \emptyset$, for $(\Lambda_{\alpha},\Lambda_{\beta})$ being either $(A,C)$ or $(B_{left},B_{right})$.
\end{enumerate}
\item given $\bar{S}=\{X_{1},\dots, X_{k}\,|\, \tilde{X}_{i}\in S,\,            Z_{j}=\tilde{X}_{k+1-j}\}$ defined in Eq.\eqref{bars}, there exist an ordered subset $ \bar{S} \supset     S^{\prime}=\{Y_{1},\dots, Y_{|S^{\prime}|}\,|\, Y_{i}\in \bar{S}\}$ where $Y_{i-1}\cap \partial Y_{i}, X_{i}\cap \partial Y_{i+1}\neq \emptyset$ that fulfill one of the following conditions:
\begin{enumerate}
    \item  $Y_{1}\cap \partial(ABC)_{1}\neq \emptyset$, $Y_{|S^{\prime}|}\cap \partial(ABC)_{2}\neq \emptyset$, 
    \item $ Y_{1}\cap \partial(\Lambda_{\alpha}) \neq \emptyset$, $Y_{|S^{\prime}|}\cap \partial(\Lambda_{\beta})\neq \emptyset$ for $(\Lambda_{\alpha},\Lambda_{\beta})$ being either $(A,C)$ or $(B_{left},B_{right})$.
\end{enumerate} 
\end{enumerate}
\end{defin}
\begin{figure}[H]
    \centering
    \resizebox{10cm}{!}{\includegraphics[scale=0.07]{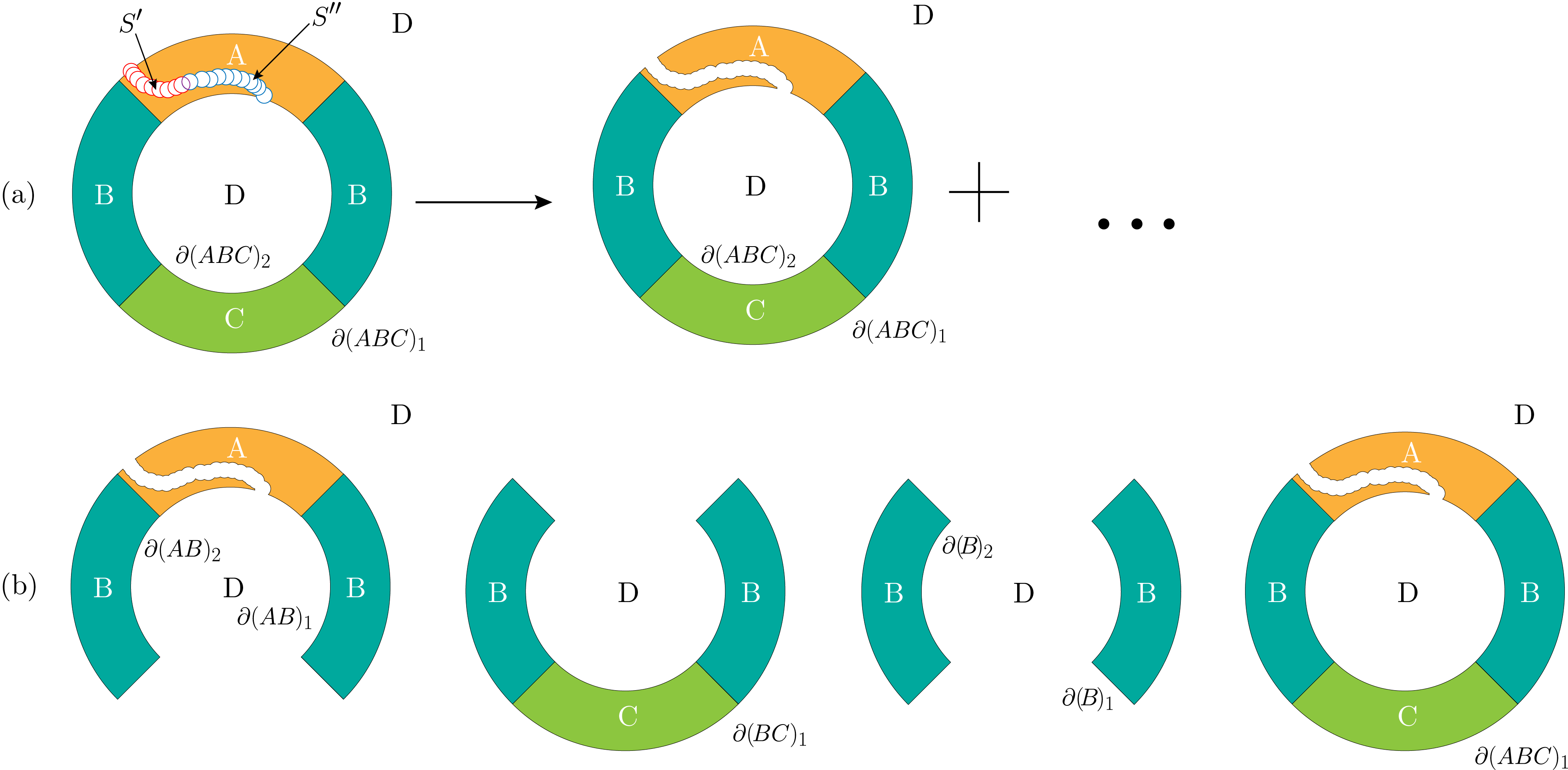}}
   \caption{An illustration for the condition $(I)$ in Definition \ref{definS}. One of the result of the union of the two subsets $S^{\prime}$ and $S^{\prime\prime}$ is the boundary modification displayed on the r.h.s of $(a)$. Such boundary modification cuts apart the donut shape of $ABC$ and changes the topological purity. To see this, define $\mathcal{Y}=\bigcup_{S^\prime} Y_{i}$ and $\mathcal{Z}=\bigcup_{S^{\prime\prime}} Z_{i}$, consider $(b)$ and note that $n_{\partial} (AB\setminus\mathcal{Y}\setminus\mathcal{Z})=n_{\partial}(B)=2$, $n_{\partial}(ABC\setminus \mathcal{Y}\setminus\mathcal{Z})=n_{\partial}(BC)=1$, therefore $\Gamma_{(AB\setminus \mathcal{Y}\setminus\mathcal{Z})}+\Gamma_{BC}-\Gamma_{B}-\Gamma_{(ABC\setminus \mathcal{Y}\setminus\mathcal{Z})}\neq 2\gamma $. Moreover, note that such boundary modification transforms $ABC$ in a simply connected region, that, as claimed, it cannot catch the topological purity, cfr. Fig.\ref{fig1} $(a)$.}
    \label{conditioni}
\end{figure}
\begin{figure}[H]
\centering
\resizebox{10cm}{!}{\includegraphics[scale=0.07]{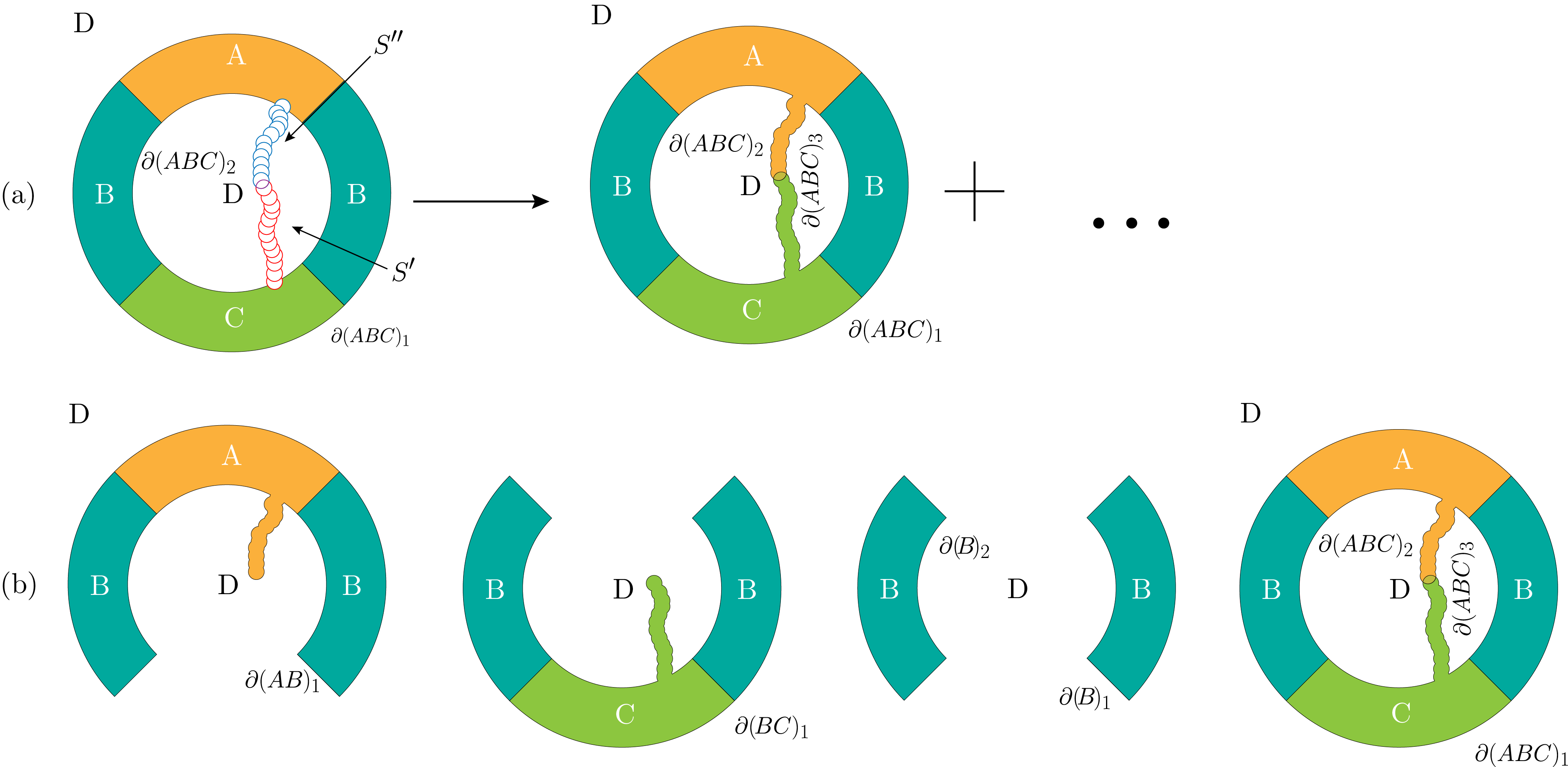}}
\caption{An illustration for the condition $(II)$ in Definition \ref{definS}. One of the result of the union of the two subsets $S^{\prime}$ and $S^{\prime\prime}$ is the boundary modification displayed on the r.h.s of $(a)$. Such boundary modification connects the regions $A$ and $C$ and therefore it changes the topological purity. To see this, define $\mathcal{Y}=\bigcup_{S^{\prime}} Y_{i}$ and $\mathcal{Z}=\bigcup_{S^{\prime\prime}} Z_{i}$, consider $(b)$ and note that $n_{\partial} (AB\cup\mathcal{Z})=n_{\partial} (BC\cup\mathcal{Y})=1$, $n_{\partial} (ABC\cup \mathcal{Y}\cup\mathcal{Z})=3$, $n_{\partial}(B)=2$, therefore $\Gamma_{(AB \cup\mathcal{Z})}+\Gamma_{(BC\cup\mathcal{Y})}-\Gamma_{B}-\Gamma_{(ABC\cup \mathcal{Y}\cup\mathcal{Z})}\neq 2\gamma $.}
\label{conditionii}
\end{figure}
To prove the main theorem, we need to prove that, for any $S\in \tilde{S}$, the topological purity keeps the value of the topological purity of the initial state, i.e. the ratio of expectation values of evolved swap operators in Eq.\eqref{toppurityforPsiS} equals the ratio of expectation values of $T_{AB}, T_{BC}, T_{B}, T_{ABC}$. In particular we consider two states, the ground state of the quantum double model and a topologically trivial state. If the initial state $\ket{\Psi_{0}}\in\mathcal{L}$ is the ground state of the quantum double model (cfr. Sec.\ref{sec:QuantumDouble} ) then $\tilde{P}_{top}(\Psi_{0}^{\otimes 2})={P}_{top}(\Psi_{0})=2^{-2\gamma}$, while if the initial state is a pure and topologically trivial state such as $\Phi$, then $ \tilde{P}_{top}(\Phi^{\otimes 2})={P}_{top}(\Phi^{\otimes 2})=1$.

\begin{proof} 
Consider $\ket{\Psi_{0}}\in\mathcal{L}$ and let us compute the topological purity of $\Psi_{S}\equiv \mathcal{R}_{S}(\Psi_{0}^{\otimes 2})$ for $S\in\tilde{\mathcal{S}}$. According to Eq.\eqref{toppurityduality} we can directly compute the expectation values of the evolution of the swap operators $T_{\Lambda}$ for $\Lambda=(AB,BC,B,ABC)$. Recalling Eq.\eqref{compact} the expectation value of the evolved swap operator $T_{\Lambda}$ is a linear combination of purities of $\Psi_{0}$ in domains $\Lambda_{\alpha}\in\mathcal{Y}_{k}(\Lambda)$: 
\be 
\aver{\mathcal{R}_{\bar{S}}(T_{\Lambda})}_{\Psi_{0}^{\otimes 2}}=\sum_{\Lambda_\alpha\in\mathcal{Y}_{k}(\Lambda)}m_{\Lambda_\alpha}P_{\Lambda_{\alpha}}(\Psi_{0})
\ee 
According to Eq.~\eqref{purity}, any purity term $P_{\Lambda_{\alpha}}(\Psi_0)$, for $\Lambda_{\alpha}\in\mathcal{Y}_{k}(\Lambda)$ equals to
\be
P_{\Lambda_\alpha}(\Psi_{0})=2^{-\log_2 |d\partial\Lambda_\alpha|+\Gamma_{\Lambda_\alpha}}
\ee
where $|\partial\Lambda_\alpha|$ is the boundary length of $\Lambda_\alpha$ and $\Gamma_{\Lambda_\alpha}$ is the corresponding topological term.
Then Eq.\eqref{toppurityduality} can be expressed as
\begin{equation}\label{Eq:mtoppure:expan}
\begin{aligned}
\tilde{P}_{top}(\Psi_S)=\frac{ \sum_{\alpha} m_{AB_\alpha} P_{AB_\alpha} \sum_\beta m_{BC_\beta} P_{BC_\beta}}{\sum_\eta m_{B_\eta} P_{B_\eta} \sum_\zeta m_{ABC_\zeta} P_{ABC_\zeta}}
\end{aligned}
\end{equation}
where we adopted a compact notation for the sum, namely $\sum_{\alpha}\equiv \sum_{AB_{\alpha}\in\mathcal{Y}_{k}(AB)}$ etc. In order to illustrate how the previous relation works, let us first make a preliminary example and consider a single domain $X$, so $S=\bar{S}=\{X\}$ and $\mathcal{R}_{S}(\cdot)=R_{X}(\cdot)$, cfr. Sec.\ref{model}; let $X$ be the domain sketched in Fig.\ref{fig11}, i.e $X\cap (\partial B)_{2}\neq 0$, we have:
\ba
\aver{R_{X}(T_{AB})}_{\Psi_{0}^{\otimes 2}}&=&N_{d}(P_{AB\cup X}+P_{AB/X})\equiv N_{d}2^{-|\partial (AB \cup X)|+\Gamma_{AB\cup X}}+N_{d}2^{-|\partial (AB / X)|+\Gamma_{AB/ X}}\nonumber\\
\aver{R_{X}(T_{BC})}_{\Psi_{0}^{\otimes 2}}&=&P_{BC}=2^{-|\partial (BC)|+\Gamma_{BC}}\nonumber\\
\aver{R_{X}(T_{B})}_{\Psi_{0}^{\otimes 2}}&=&N_{d}(P_{B\cup X}+P_{B/X})\equiv N_{d}2^{-|\partial (B \cup X)|+\Gamma_{B\cup X}}+N_{d}2^{-|\partial (B / X)|+\Gamma_{B/ X}}\nonumber\\
\aver{R_{X}(T_{ABC})}_{\Psi_{0}^{\otimes 2}}&=&P_{ABC}=2^{-|\partial (ABC)|+\Gamma_{ABC}}
\ea
the topological purity is given by their ratio:
\be
\tilde P_{top}^{(1)}=\frac{N_{d}(2^{-|\partial (AB \cup X)|-|\partial (BC)|+\Gamma_{AB\cup X}+\Gamma_{BC}}+2^{-|\partial (AB / X)|-|\partial (BC)|+\Gamma_{AB/ X}+\Gamma_{BC}})}{N_{d}(2^{-|\partial (B \cup X)|-|\partial (ABC)|+\Gamma_{B\cup X}+\Gamma_{ABC}}+2^{-|\partial (B / X)|-|\partial (ABC)|+\Gamma_{B/ X}+\Gamma_{ABC}})}=2^{-2\gamma}
\ee
where we used the following relations:
\ba
|\partial (AB \cup X)|+|\partial (BC)|&=&|\partial (B \cup X)|+|\partial (ABC)| \nonumber\\
|\partial (AB / X)|+|\partial (BC)|&=&|\partial (B / X)|+|\partial (ABC)| \nonumber\\
\Gamma_{AB\cup X}+\Gamma_{BC}-\Gamma_{B\cup X}+\Gamma_{ABC}&=&-2\gamma\nonumber\\
\Gamma_{AB/ X}+\Gamma_{BC}-\Gamma_{B/ X}+\Gamma_{ABC}&=&-2\gamma\nonumber\\
\tilde P_{top}(\Psi_{0}^{\otimes 2})&=&2^{-2\gamma}
\ea
In fact, for any two terms $m_{B_\eta} P_{B_\eta}$ and $m_{ABC_\zeta} P_{ABC_\zeta}$ chosen from the denominator of Eq~\eqref{Eq:mtoppure:expan}, one can always find two one-to-one corresponding terms $m_{AB_\alpha} P_{AB_\alpha}$ and $m_{BC_\beta} P_{BC_\beta}$ from the numerator, following the rule that the boundaries of them share the same trajectory of boundary statistic dynamics (see Fig~\ref{fig11} for the concrete rule). 
\begin{figure}[H]
\centering
\resizebox{10cm}{!}{\includegraphics[scale=0.6]{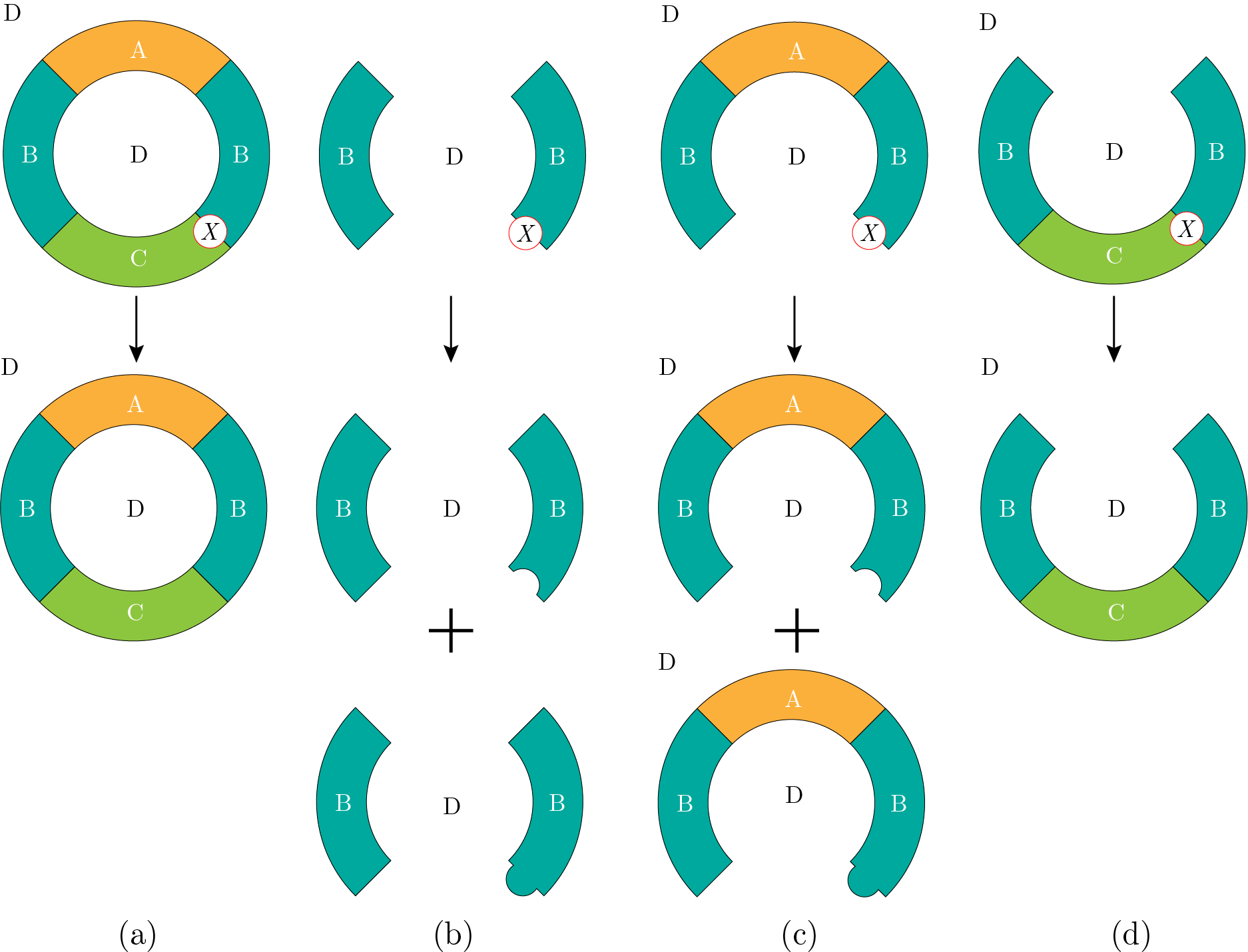}}
\caption{An illustration for the single domain noise $R_{X}(\cdot)$. The figure is a pictorial representation of the action of $R_{X}(\cdot)$ for $X\cap \partial B\neq \emptyset$ and $X\cap \partial C\neq \emptyset$ on $(a)$ $T_{ABC}$, $(b)$ $T_{B}$, $(c)$ $T_{AB}$, $(d)$ $T_{BC}$. }
\label{fig11}
\end{figure}
As a consequence, it is not hard to find that the boundary length in purity expressions can cancel with each other $\partial AB_\alpha+\partial BC_\beta=\partial B_\eta+\partial ABC_\zeta$, and the corresponding accumulative coefficients fulfill $m_{AB_\alpha} m_{BC_\beta}=m_{B_\eta} m_{ABC_\zeta}$. In union with the fact that the topological terms keep the relation $(\Gamma_{B_\eta}+\Gamma_{ABC_\zeta})-(\Gamma_{AB_\alpha}+\Gamma_{BC_\beta})=2\gamma$ if $S\in \tilde{\mathcal{S}}$ since the perturbation can not cut apart the donut shape of system $ABC$, we can reach the conclusion that
\be
\tilde P_{top}(\Psi_{S})=\tilde P_{top}(\Psi_{0}^{\otimes 2}), \quad \forall S \in \tilde{\mathcal{S}}
\ee
The proof for $\Phi$ being a pure and topologically trivial state is identical to the one presented above, with the only difference that $P_{\Lambda_{\alpha}}(\Phi)=\aver{T_{\Lambda_{\alpha}}}_{\Phi^{\otimes 2}}=1$ for any $\Lambda_{\alpha}\in\mathcal{Y}_{k}(\Lambda)$, cfr. Eq.\eqref{compact}.  This concludes the proof.
\end{proof}
\begin{remark*}
One wonders whether there are other boundary modifications that destroy the donut shape that can change the topological purity other than the one of conditions $(I)$ and $(II)$. The answer is no: consider for example the ordered sets of domains $S^{\prime}$ and $S^{\prime\prime}$ in Fig.\ref{sasa} $(a)$. One of the result of their application is the boundary modification displayed in Fig.\ref{sasa} $(b)$. Although it destroys the donut shape of $ABC$ it does not change the topological purity, indeed $n_{\partial}(AB)=n_{\partial}(B)=2$, $n_{\partial}(BC)=1$, $n_{\partial}(ABC)=3$ and therefore  $\Gamma_{(AB\cup\mathcal{Z})}+\Gamma_{(BC\cup\mathcal{Y})}-\Gamma_{B\cup\mathcal{Y}}-\Gamma_{(ABC\cup \mathcal{Y}\cup\mathcal{Z})}\neq 2\gamma $.
\end{remark*}
\begin{figure}[H]
\centering
\resizebox{10cm}{!}{\includegraphics[scale=0.6]{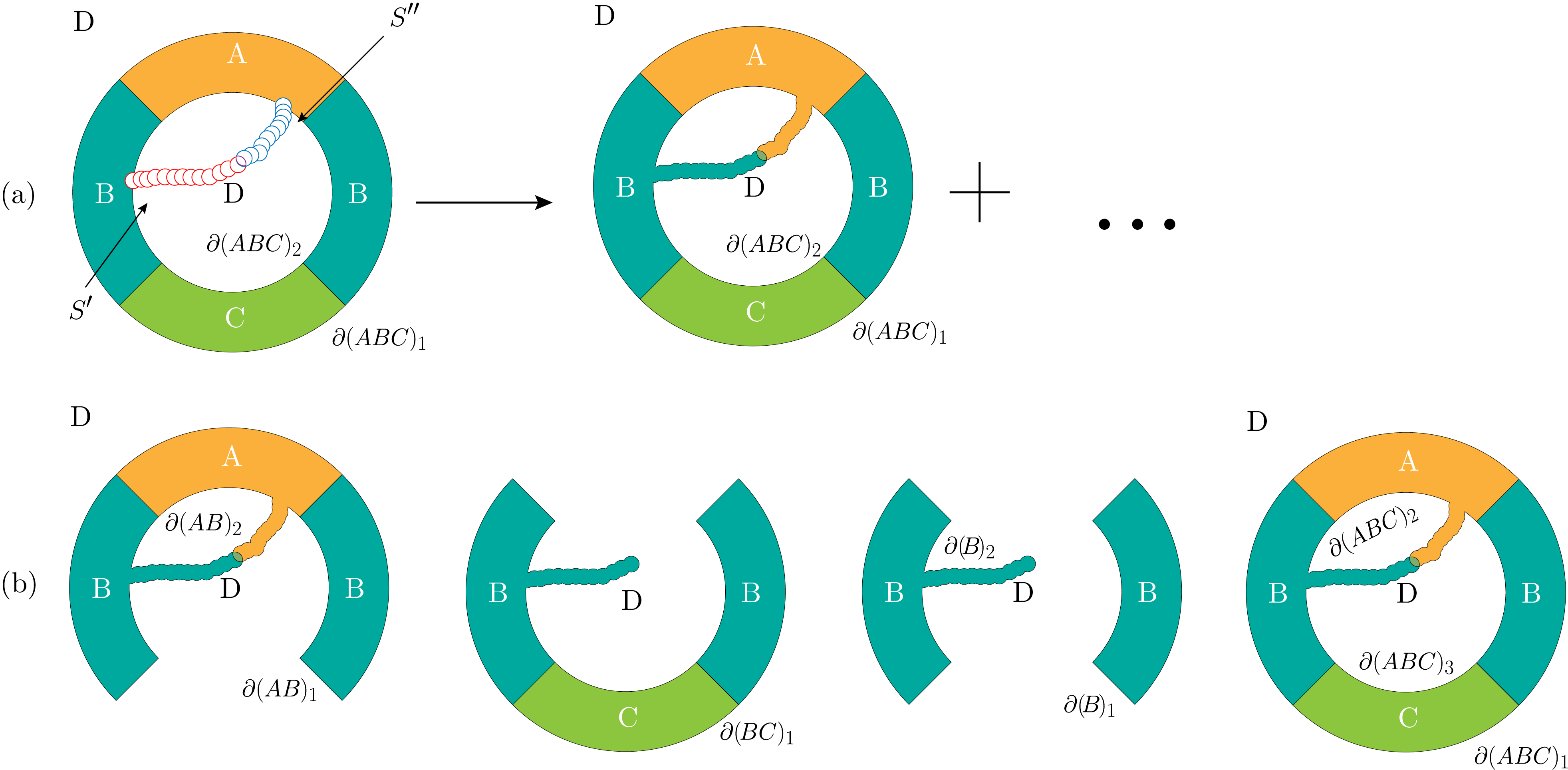}}
\caption{An illustration of a boundary modification that destroys the donut shape without changing the topological purity. In more details, $(a)$ the application of the random circuit described by the two ordered subsets $S^{\prime}$ and $S^{\prime\prime}$ give rise, in the worst-case scenario, to an intersection between the internal boundary of $A$ and the internal boundary of $B$. Although this operation breaks the donut shape, the topological purity keeps steady, to see this define $\mathcal{Y}=\bigcup_{S^\prime} Y_{i}$ and $\mathcal{Z}=\bigcup_{S^{\prime\prime}} Z_{i}$, consider $(b)$ and note that $n_{\partial} (AB\cup\mathcal{Y}\cup\mathcal{Z})=n_{\partial}(B\cup\mathcal{Y})=2$, $n_{\partial\cup\mathcal{Y}}(BC)=1$, $n_{\partial}(ABC\cup\mathcal{Y}\cup\mathcal{Z})=3$ and thus $\Gamma_{(AB\cup\mathcal{Y}\mathcal{Z})}+\Gamma_{(BC\cup\mathcal{Y})}-\Gamma_{B\cup\mathcal{Y}}-\Gamma_{(ABC\cup \mathcal{Y}\cup\mathcal{Z})}= 2\gamma $.}
\label{sasa}
\end{figure}
\section{Conclusions}
In this paper, we addressed some questions regarding the stability of topological order under noisy perturbations, but the path to find a general analytic proof is still long and tortuous. Working with the ground state $\Psi_0$ of quantum double models, we defined a new detector for topological order - the topological purity $P_{top}$ - proving its robustness in two distinct phases, namely the topological phase and the trivial phase. More precisely, as a noise model, we introduced a set of quantum maps that mimics the evolution of local random quantum circuits. The two phases are indeed created by the quantum maps as orbits of two initially distinct states, the ground state of quantum double models $\Psi_0$ and a pure, topologically trivial state. We found that the topological purity attains two different constant values among such states, in particular $P_{top}<1$ for the topologically ordered phase and $P_{top}=1$ for the trivial phase. The dynamics of the topological purity under such noise model can be mapped onto dynamics for the subsystems used to define the topological purity, cfr. Sec.\ref{sec:PT}. This peculiarity enabled us to prove our main theorem and to provide many pictorial representations, giving the reader more intuition on the effects of the noisy dynamics.
Despite the generality of the setup, our result is not the final word regarding the stability of topological order, and even more general and complete proofs are necessary to go further in this direction. This paper opens a series of different questions that might be interesting to investigate. 
Our proof strongly relies on the fact that the purity of the ground state of the quantum double model is a function of the boundary only: is this the case for other models and, more generally, is our proof working for any topologically ordered state? So far we only investigated the stability of topological purity, strictly related to the order two R\'enyi entropy: it would be indeed interesting to find an analytic treatment for more general detectors of topological order. 
\section*{Acknowledgments} We acknowledge support from NSF award number 2014000. 
%%%%%%%%%%%%%%%%%%%%%%%%%%%%%%%%%%%%%%%%

\end{document}